%% file: ITW2010.tex
\documentclass[10pt,conference]{IEEEtran}
\usepackage[centertags]{amsmath}
\usepackage{amsfonts}
\usepackage{amssymb}
\usepackage{epsfig}
\usepackage{amsmath,  amssymb}
\usepackage{graphicx}
\usepackage[centertags]{amsmath}
\usepackage{clrscode}
 
\usepackage{cite}

\usepackage{fullpage}
\usepackage{amsmath,amssymb,pseudocode}
\usepackage{color}
\usepackage[normalem]{ulem}

\renewcommand{\baselinestretch}{1}
\begin{document}
\title{On Coding for  Cooperative  Data Exchange }

\author{\authorblockN{Salim El Rouayheb}
\authorblockA{Texas A\&M University \\
Email: rouayheb@tamu.edu}
\and
\authorblockN{Alex Sprintson}
\authorblockA{Texas A\&M University \\
Email: spalex@tamu.edu}
\and
\authorblockN{Parastoo Sadeghi}
\authorblockA{Australian National University\\
Email: parastoo.sadeghi@anu.edu.au}
\and
}

\long\def\symbolfootnote[#1]#2{\begingroup%
\def\thefootnote{\fnsymbol{footnote}}\footnote[#1]{#2}\endgroup}

\maketitle

\newtheorem{theorem}{Theorem}
\newtheorem{lemma}[theorem]{Lemma}
\newtheorem{claim}[theorem]{Claim}
\newtheorem{proposition}[theorem]{Proposition}
\newtheorem{condition}[theorem]{Condition}
\newtheorem{observation}[theorem]{Observation}
\newtheorem{conjecture}[theorem]{Conjecture}
\newtheorem{property}[theorem]{Property}
\newtheorem{assertion}[theorem]{Assertion}
\newtheorem{example}[theorem]{Example}
\newtheorem{definition}[theorem]{Definition}
\newtheorem{corollary}[theorem]{Corollary}
\newtheorem{remark}[theorem]{Remark}
\newtheorem{note}[theorem]{Note}
\newtheorem{problem}[theorem]{Problem}

\input{ITW2010-abstract.tex}


 \symbolfootnote[0]{The work of Alex Sprintson was supported  by Qatar Telecom (Qtel), Doha, Qatar.\\
 The work of Parastoo Sadeghi was supported under Australian Research
Council's Discovery Projects funding scheme (project no. DP0984950).}

\input{ITW2010-introduction.tex}

\input{ITW2010-model.tex}

\input{ITW2010-bounds.tex}

\input{ITW2010-alg2.tex}

\input{ITW2010-rank.tex}

\input{ITW2010-numerical.tex}

\input{ITW2010-conclusion.tex}



%
%

\bibliographystyle{IEEEtran}
\bibliography{IEEEabrv,refs2009}

\end{document}

%% file: ITW2010-abstract.tex
\begin{abstract}
We consider the problem of data exchange by a group of closely-located wireless nodes. In this problem each node holds a set of packets and needs to obtain all the packets held by other nodes. Each of the nodes can
broadcast the packets in its possession (or a combination thereof) via a noiseless broadcast channel of capacity one packet per channel use. The goal is to minimize the total number of transmissions needed to satisfy the demands of all the nodes, assuming that they can cooperate with each other and are fully aware of the packet sets available to other nodes. This problem arises in several practical settings, such as peer-to-peer systems and wireless data broadcast. In this paper, we establish upper and lower bounds on the optimal number of transmissions and present an efficient algorithm with provable performance guarantees. The effectiveness of our algorithms is established through numerical simulations.
\end{abstract} 

%% file: ITW2010-introduction.tex
\section{Introduction}

In recent years there has been a growing interest in developing cooperative strategies for wireless communications~\cite{kramer:guest:tit:2007,sendonaris:erkip:aazhang:tcom:03:part:i}. Cooperative communication is a promising technology for the future that can provide distributed space-time diversity, energy efficiency, increased coverage, and enhanced data rates. In a cooperative setting users aid each other to achieve a common goal sooner, or in a more robust or energy efficient manner. This is in contrast to a traditional non-cooperative scenario where users compete against each other for channel resources (time, bandwidth, \emph{etc}).
 

In this paper we consider the problem of cooperative data exchange. To motivate the problem, consider a group of mobile users or clients who wish to download a large file, which is divided into $n$ packets, from a base station. The common goal here is to minimize the total download time. The long-range link from each client to the base station is subject to long-term path loss and shadowing, as well as short-term fading~\cite{Rappaport2002}, which often render it unreliable and slow. As a result, after $n$ transmissions from the base station, each client may have received only a subset of degrees of freedom required for full decoding. A possible strategy for the base station is to employ network coding and keep transmitting innovative packets until every client has fully decoded all packets~\cite{sadeghi:traskov:koetter:netcod:09,sundararajan:shah:medard:isit:08}.

\begin{figure}[b]
\begin{center}
\includegraphics[width=2.5in]{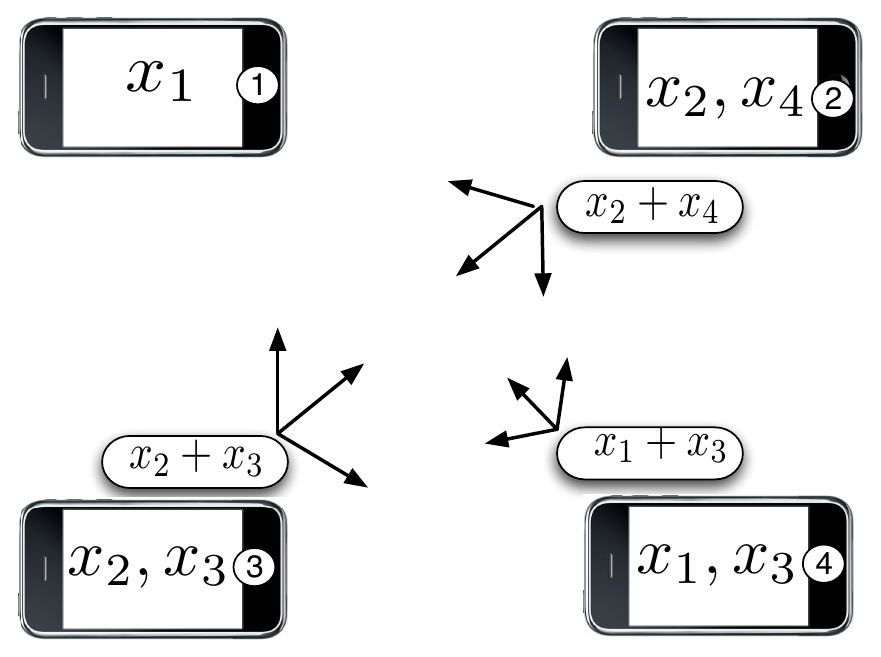}
 \caption{Data exchange among four clients. \label{fig:example1}}%
\end{center}
\end{figure}

If the clients happen to be in the vicinity of each other an alternative strategy for the clients is to switch to short-range transmission as soon as all packets are collectively owned by the group. This strategy has two main benefits:
\begin{itemize}
	\item  short-range communications is often much more reliable and faster;
	\item  after only $n$ transmissions, the valuable long-range channel is freed and the base station can serve other clients in the system.
\end{itemize}

To illustrate the problem further, consider four wireless clients who had requested $n=4$ packets, $x_1,\dots,x_4 \in GF(2^m)$, from the base station. However, due to channel imperfections, the first client only received packet $x_1$, while second, third, and fourth clients received packets $\{x_2, x_4\}$, $\{x_2, x_3\}$, and $\{x_1, x_3\}$, respectively. Since they have collectively received all the packets, they can now try to communicate among themselves to complete the communication and  ensure that all the clients eventually possess all the packets. 

In this paper we make the following assumptions:
\begin{enumerate}
	\item each mobile client can broadcast data to all other clients at a rate of one packet per transmission, i.e., $m$ bits per transmission in this example. Furthermore, all the clients  receive this transmission error-free;
	\item  each client knows the  packets that were received by others.
\end{enumerate}

The first assumption is justified  when mobile users are close to each other. The second condition can be easily achieved at the beginning by broadcasting the index of the received packets to the other clients.

A naive solution to this problem would require four transmissions: the first client sends $x_1$, the second client sends $x_2$ and then $x_4$, and finally the third client sends $x_3$.
But, the question that we are interested in here is ``can the clients do better in terms of number of transmissions and if so, how?'' For the above example, it is easy to see that with coding, we can reduce the number of transmissions from 4 to 3.  Figure~\ref{fig:example1} shows a coding scheme with 3 transmissions where the second, third, and fourth clients  send the coded packets $x_2+x_4$, $x_2+x_3$ and $x_1+x_3$, respectively. It can be verified that all the mobile clients can then decode all the packets. In fact, this is the minimum number of transmissions since the first client had initially received one packet hence it needs to receive at least three degrees of freedom from the other clients.

This problem is  related to the \emph{index coding} problem \cite{BK06,BBJK06, RSG08,RSG09ISIT,RSA09i} in which the different clients cannot communicate with each other, but can receive transmissions from a server possessing all the data. Moreover, in the index coding problem different clients might have different demands, while in the problem considered here each client wants to obtain  all the available packets. Another related line of work is that of gossip algorithms \cite{D09} where the goal is to efficiently and distributively compute a function of the data present in a dynamic network.


\textbf{Contribution.} In this paper, we develop a framework towards finding optimum strategies for cooperative data exchange. To the best of our knowledge, this problem setting has not been previously considered in the literature. We establish upper and lower bounds on the optimal number of transmissions. We also present an efficient algorithm with a provable performance guarantee. The effectiveness of our approach is verified through numerical simulations.

\textbf{Organization.} The rest of this paper is organized as follows.
In Section~\ref{sec:model}, we present the model and formally define the problem.
In Section~\ref{sec:bounds},  we establish lower and upper bounds on the required number of transmissions. In Section~\ref{sec:alg}, we present a deterministic algorithm and analyze its performance. Our evaluation results are presented in Section~\ref{sec:numerical}. Finally, conclusions and directions for future work appear in Section~\ref{sec:conclusion}.


%% file: ITW2010-model.tex
\section{Model}\label{sec:model}

The problem can be formally defined as follows. A set $X=\{x_1,\dots,x_n\}$ of $n$ packets each belonging to a finite alphabet $\mathcal{A}$ needs to be delivered to a set of $k$ clients $C=\{c_1,\dots,c_k\}$. Each client $c_i\in C$ initially holds a subset $X_i$ of packets in $X$, i.e., $X_i\subseteq X$. We denote by $n_i=|X_i|$ the number of packets initially available to  client $c_i$, and by $\overline{X_i}=X\setminus X_i$ the set of packets required by $c_i$.   We assume that the clients collectively know all packets in $X$, i.e., $\displaystyle\cup_{c_i\in C} X_i = X.$ Each client can communicate to all its peers through an error-free broadcast channel of capacity one packet per channel use. The problem is to find a scheme that requires the minimum number of transmissions to make all the packets available to all the clients.

We focus in this paper on the design of linear solutions to the problem at hand. In a linear solution, each packet is considered to be an element of a finite field $\mathbb{F}$ and all the encoding operations are linear over this field. For a given instance of this problem, define $\tau$ to be the minimum total  number of transmissions required to satisfy the demands of all the receivers with linear coding.

We denote by $n_{min}$ the minimum number of packets held by a client, i.e.,
$n_{min} =\min_{1\leq i\leq k}n_i,$ and by $n_{max}$ the maximum number of packets held be a client, i.e.,  $n_{max} =\max_{1\leq i\leq k}n_i$.

We say that a client $c_i$ has a \emph{unique} packet $x_j$ if \mbox{$x_j\in X_i$} and $x_j\notin X_\ell$ for all $\ell\neq i$. Note that, without loss of generality, we can assume that no client has a unique packet. Indeed, a unique packet can be broadcast by the client holding it in an uncoded form without any penalty in terms of optimality.

%
%
%
%


%% file: ITW2010-bounds.tex
\section{Upper and Lower Bounds}\label{sec:bounds}

We begin by establishing a lower bound on the number of transmissions.

\begin{lemma}\label{lem:1}
The minimum number of transmissions  is greater or equal to $n-n_{min}$. If all clients initially have the same number of packets $n_{min}< n$, i.e., $n_i=n_{min}$ for  $i=1,\dots, k$, then the minimum number of transmissions  is greater or equal $n-n_{min}+1$.
\end{lemma}
\begin{proof}
 The first part follows from the fact that each client needs to receive at least $n-n_i$ packets. The second part follows from the fact that a transmitting client does not benefit from its own transmissions.
\end{proof}

Next, we present an upper bound on the minimum required number of
transmissions $\tau$.

\begin{lemma} \label{lem:UpperBound1}
 For $|\mathbb{F}|\geq k$, it holds that
\begin{align}\label{eq:UpperBound}
    \tau&\leq \min_{1\leq i\leq k}\{| \overline{X}_i|+\max_{1\leq j\leq k}|\overline{X}_j\cap X_i|\}.
\end{align}
\end{lemma}

\begin{proof}
Consider the following solution consisting of two phases:
\begin{enumerate}
  \item Phase 1: pick a client $c_i$ and make it a ``leader'' by satisfying all of its  demands with
    uncoded transmissions from the other clients. This requires $|\overline{X_i}|$ transmissions.
  \item Phase 2: client $c_i$ broadcasts coded packets to satisfy the demands of  all the other clients.
\end{enumerate}

 After Phase 1, each client $c_j$ knows all the packets in $X_j\cup \overline{X_i}$ and, thus, requires packets
  $\overline{X}_j\cap X_i$. By using  network coding techniques (see e.g., \cite{JSCEEJT04}), Phase 2 can
   be accomplished using $\max_{j}|\overline{X}_j\cap X_i|$ transmissions, provided that the size of the finite field
   $\mathbb{F}$ is at least $k$. Indeed, the leader $c_i$ can form encoded packets in such a way that, after each transmission,  the degree
    of freedom is increased by one  for every client $j\neq i$ who still have not received all the packets.
\end{proof}

The bound of Lemma \ref{lem:UpperBound1} is tight in many instances, in particular when all the sets $X_i$ are disjoint (i.e., \mbox{$X_i\cap X_j=\emptyset, i\neq j$}). In this trivial case, the minimum number of transmissions is equal to $n$, and this bound is tight. The following is a non-trivial example where the above bound is
also tight: $X_1=\{x_2,x_3\}$, $X_2=\{x_1,x_3\}$ and
$X_3=\{x_1,x_2\}$. In this case, the upper bound of Lemma \label{lem:UpperBound} gives $\tau\leq
2$. But, lemma \ref{lem:1} gives $\tau\geq n-n_{min}+1=2$. Therefore $\tau=2$, and one way this can be achieved is by letting $c_1$ and $c_2$ transmit $x_2+x_3$ and $x_1+x_3$, respectively.


The bound of Lemma \ref{lem:UpperBound1} is not always tight since
the scheme described in the proof
 is not guaranteed to be optimal. This is due to the fact that a client is
 made a leader by only transmitting  uncoded messages. Consider, for example, the following instance
 with four clients and where $X_1=\{x_2,x_3,x_4\}$, $X_2=\{x_1,x_4\}$, $X_3=\{x_1,x_2,x_4\}$ and $X_4=\{x_1,x_3\}$.
 From the previous lemma, we get $\tau\leq 3$. However, there exists a solution with two transmissions where  $c_1$
 transmits $x_3+x_4$ and $c_3$ transmits  $x_1+x_2+x_4$. By Lemma \ref{lem:1}, we know that this scheme is optimal,
 and $\tau=2$.

In the next section, we present an additional bound on $\tau$. In particular, Lemma~\ref{lem:alex:alg} shows that $$\tau\leq 2n-n_{max}-n_{min}.$$

%% file: ITW2010-alg2.tex
\section{A Deterministic Algorithm}\label{sec:alg}

We proceed to present an efficient deterministic algorithm for the information exchange problem.
At each iteration of the algorithm, one of the clients broadcasts a linear combination of the packets in $X$.
 For a  coded packet $x$ we denote by $\Gamma_x\in \mathbb{F}^n$ the corresponding vector of linear coefficients,
  i.e., $x=\Gamma_x\cdot (x_1,\dots,x_n)^T$.

We also denote by $Y_i$ the subspace spanned by  vectors  corresponding to the linear combinations available
 at client $c_i$.  In the beginning of the algorithm, $Y_i$ is equal to the subspace spanned by vectors that correspond
  to the packets in  $X_i$, i.e., $Y_i=\langle\{\Gamma_x\ |\ x\in X_i\}\rangle$. The goal of our algorithm is to
   simultaneously increase the dimension of the subspaces $Y_i$, $i=1,\dots, k$,  for as many clients as possible.


Specifically, at each iteration, the algorithm identifies a client $c_i\in C$ whose subspace $Y_i$ is of maximum dimension.
 Then, client $c_i$ selects a vector $b \in Y_i$ in a way that increases the dimension of $Y_j$ for each client
 $c_j\neq c_i$, and transmits the corresponding packet $b\cdot(x_1,\dots,x_n)^T$. Vector $b$ must satisfy $b\notin Y_j$ for all  $j\neq i$.
 Such a vector $b$ exists and can be selected using the network coding techniques \cite{JSCEEJT04} provided that
  $|\mathbb{F}|\geq k$.

At any iteration, the subspace associated with a certain client will correspond to the original packets possessed by this
client, in addition to the transmitted packets in the previous transmissions. As a result, at some iteration, the subspaces
 associated with a number of clients may become identical. In this case, without loss of generality, we merge this group
 of clients into a single client with the same subspace.

 The formal description of the algorithm is presented below.

%

\begin{codebox}
\Procname{$\proc{Algorithm\ IE }(Information\ Exchange)$}
\li \For $i \gets 1$ \To $k$
\li \Do
\li  $Y_i=\langle \{\Gamma_x\ |\ x\in X_i\}\rangle$
\End
\li \While there is a client $i$ with $\dim Y_i< n$ \label{ref:alg:2:step:1}
\li \Do
\li \While $\exists c_i, c_j\in C$  $i\neq j$, such that $Y_i=Y_j$
\li \Do
\li  $C=C\setminus \{c_i\}$\label{ref:alg:2:step:2}
\End
\li  Find a client $c_i$ with a   subspace $Y_i$ of
\zi  maximum dimension \label{ref:alg:2:step:3} (If there are multiple
\zi such clients  choose an arbitrary one of them)
\li  Select a vector $b\in Y_i$ such that $b\notin Y_j$
\zi  for each $j\neq i$ \label{ref:alg:2:step:4}
\li  Let client $c_i$  broadcast packet $b\cdot(x_1,\dots,x_n)^T$
\li \For $\ell=1 \gets 1$ \To $k$
\li \Do
\li  $Y_i\leftarrow Y_i +  \langle \{b\}\rangle$
\End


\end{codebox}

%
%

\begin{lemma}\label{lem:alex:alg}
     The number of transmissions made by $\proc{Algorithm\ IE}$ is at most $\min\{n,2n-n_{max}-n_{min}\}$,  provided that
     $|\mathbb{F}| \geq k$.
\end{lemma}
\begin{proof}\label{pf:1}
    First, we show that in Step~\ref{ref:alg:2:step:4} of the algorithm it is always possible to select a vector $b$
     in $Y_i$ such that $b\notin Y_j$ for each $j\neq i$. We note that the algorithm maintains the invariant $Y_i\neq Y_j$
      for $i\neq j$. Since in Step~\ref{ref:alg:2:step:3} of the algorithm we select a client $c_i$ with a maximum
      dimension of $Y_i$, we can then use an    argument similar to that used in \cite{JSCEEJT04} to show that there exists a vector $b\in Y_i$ such that
       $b\notin Y_j$ for $j\neq i$.

Note also that once $Y_i=Y_j$ for two different clients $c_i$ and $c_j$, they will be identical for the rest of the
algorithm. Thus, we can remove one of the clients at Step~\ref{ref:alg:2:step:2} as described previously.

We proceed to analyze the number of transmissions made by $\proc{Algorithm\ IE}$. We note that  each transmission
is linearly independent of the others. Therefore, the total number of transmission is bounded by $n$.

%
%
   Let $c_i$ be a client with $|X_i|=n_{max}$. Note that $n-n_{max}$ transmissions
   are needed in order to satisfy the demands of $c_i$. We consider two cases:

    First,
suppose that $\dim Y_i \neq n$ until the last iteration of the main
loop (that begins on Step 4). Then, let $c_j$ be a client that
has transmitted a packet at the last iteration. Note that
at the beginning of the last iteration $Y_i\subset Y_j$ since $c_j$
should know all the packets to be able to transmit in the
last iteration. Therefore, $c_i$ and $c_j$ will only merge upon the
completion of the algorithm. In terms of the number of
transmissions, the worst case scenario occurs when in
each transmission round, the dimensionality of either $Y_i$
or $Y_j$ , but not both, increases by one. Since $|X_j| > n_{min}$,
the total number of transmissions is at most $n-n_{max}+
n - n_{min} = 2n-n_{max}-n_{min}$.

Second, suppose that the dimension of $Y_i$ is equal to $n$ before the last iteration of the main loop. In this case,
 we select $c_j$ to be the client  for which $|Y_j|\neq n$ until the last iteration. Using a similar argument as in the first case, we can show that the
 number of transmissions is at most  $2n-n_{max}-n_{min}$.
\end{proof}

\begin{corollary}
 The number of transmissions made by $\proc{Algorithm\ IE}$ is at most two times more than the optimal algorithm.
\end{corollary}
\begin{proof}
    By Lemma~\ref{lem:alex:alg}, the number of transmission made by $\proc{Algorithm\ IE}$ is at
    most $2n-n_{max}-n_{min}$. By Lemma~\ref{lem:1}, the optimum number number of transmission is at
     least $n-n_{min}$. Since $n_{max}\geq n_{min}$ it holds the number of transmission is at most twice the optimum.
\end{proof}

Our empirical results, presented in Section~\ref{sec:numerical}, show that $\proc{Algorithm\ IE}$ performs very well in practical settings.

%% file: ITW2010-rank.tex
\section{Relations to Rank Optimization Problems}

In this section, we show that $\tau$ can be obtained by solving a rank optimization problem.

Let $n'=|X_1|+|X_2|+\dots+|X_k|=n_1+n_2+\dots+n_k$. We define $\mathbb{M}_\mathbb{F}(n_i,n)$ to be the set
 of $n_i\times n$ matrices with entries in the field $\mathbb{F}$. To each client $c_i$, we associate the
 set of matrices
\begin{equation*}
\begin{split}
\mathbb{A}_\mathbb{F}^i &:=\{[a_{hl}]\in
\mathbb{M}_\mathbb{F}(n_i,n)|a_{hl}=0 \text{ if }
 x_l\notin X_i,\\
 &\quad \ell=1,\dots,n_i \}.
 \end{split}
 \end{equation*}
We also define the set of matrices
$$
\mathbb{A}_\mathbb{F} :=\{A\in \mathbb{M}_\mathbb{F}(n',n)|
 A =\begin{bmatrix}
    \underline{A_1 } \\
    \underline{A_2}\\
    \vdots\\
   \overline{ A_k}
  \end{bmatrix},\text{ where } A_i\in\mathbb{A}_\mathbb{F}^i \}.
$$

For example, the matrices $A_i\in\mathbb{A}_\mathbb{F}^i$ for the instance depicted in Figure~\ref{fig:example1} have the following form:
\begin{equation*}
 A_1 =\begin{bmatrix}
    *&0&0&0  \\
  \end{bmatrix},
   A_2 =\begin{bmatrix}
    0&*&0&*  \\
    0&*&0&*
  \end{bmatrix}, \end{equation*}

  \begin{equation*}
  A_3 =\begin{bmatrix}
    0&*&*&0 \\
    0&*&*&0
  \end{bmatrix},
  A_4 =\begin{bmatrix}
    *&0&*&0  \\
    *&0&*&0
  \end{bmatrix}.
 \end{equation*}
where the ``*'' is the ``don't care'' symbol, each entry with this
symbol can independently take any value in the field $\mathbb{F}$.

Let $e_1,e_2,\dots, e_n$ be the canonical basis of the vector space
$\mathbb{F}^n$, i.e.\, the coordinates of $e_i$ are all zeros except
the $i$-th coordinate which is equal to 1. To each client $c_i$, we
also associate the matrix $B_i\in\mathbb{M}_\mathbb{F}(n_i,n)$ whose
row vectors are vectors $e_j$ in the canonical basis satisfying
$x_j\in X_i$.

Going back to the   instance depicted in Figure~\ref{fig:example1}, we have:
\begin{equation*}
 B_1 =\begin{bmatrix}
    1&0&0&0  \\
  \end{bmatrix},
   B_2 =\begin{bmatrix}
    0&1&0&0  \\
    0&0&0&1
  \end{bmatrix}, \end{equation*}

  \begin{equation*}
  B_3 =\begin{bmatrix}
    0&1&0&0 \\
    0&0&1&0
  \end{bmatrix},
  B_4 =\begin{bmatrix}
    1&0&0&0  \\
    0&0&1&0
  \end{bmatrix}.
 \end{equation*}


The following theorem is easy to establish:

\begin{theorem}
The minimum number of transmissions $\tau$ achieved by linear codes is given by  the following optimization problem:

$$\tau=\min_{A\in \mathbb{A}_\mathbb{F}} \text{rank }(A)$$ subject to:
\begin{equation*}
 \text{rank } \left(\begin{bmatrix}
    A \\
    B_i
  \end{bmatrix}\right)= n,\quad \forall i=1,\dots,k.
 \end{equation*}

\end{theorem}


%% file: ITW2010-numerical.tex
\section{Numerical Results}\label{sec:numerical}
In this section, we evaluate the lower and upper bounds on the optimum number of transmissions $\tau$. We also verify the performance of the algorithm presented in the previous section.

Figure~\ref{fig:itw} shows numerical results for $k=3$ clients\footnote{Note that in our numerical analysis, we have taken into account the number of transmissions for broadcasting unique packets. Since this number is the same for all schemes, the relative performance is unaffected.}. The number of packets ranges from $n=10$ to $n=50$. Each curve (except for the top one) represents the average over 100 random initializations of the problem (we randomly selected $X_i$ subject to $\cup X_i = X$). The bottom curve shows the lower bound of Lemma~\ref{lem:1}. The next curve is average number of transmissions required by $\proc{Algorithm\ IE}$. Remarkably, the algorithm performs very close to lower bound. The next curve shows the upper bound of Lemma~\ref{lem:UpperBound}. Finally, the top curve shows the trivial upper bound of $n$ transmissions.

Finally, we have observed similar trends for a larger number of clients $k$ and have omitted the numerical results for brevity.

\begin{figure}
\begin{center}
  \includegraphics[width=\columnwidth]{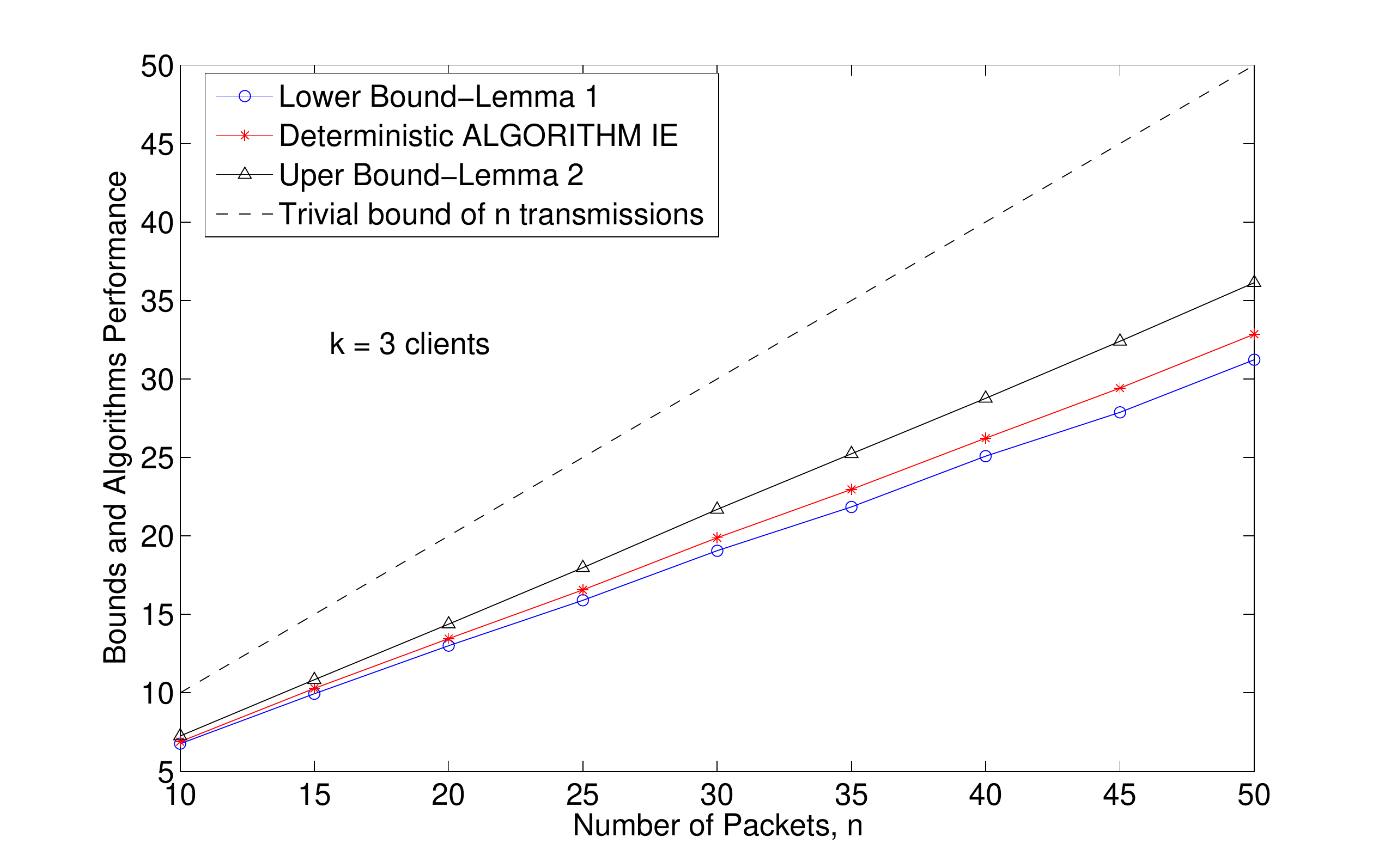}
  \caption{Numerical results with three clients using $\proc{Algorithm\ IE}$ and its comparison with the  upper and lower bounds of Section~\ref{sec:bounds}.}\label{fig:itw}
  \end{center}
\end{figure}

%% file: ITW2010-conclusion.tex
\section{Conclusion}\label{sec:conclusion}
In this paper, we considered the problem  of cooperative data
exchange by a group of wireless clients. Our figure of merit was the
number of transmissions to ensure that each client eventually
obtains all the data. We have established upper and lower bounds on
the optimum number of transmissions. We also presented a
deterministic algorithm, referred to as $\proc{Algorithm\ IE}$, for this problem and analyzed its performance. Empirical results pointed that  the proposed algorithm
performs remarkably close to the lower bound.

This work was only a first step towards understanding  the problem
and there are many interesting directions for future research. We
still do not know  the optimal solution or the computational
complexity of finding one. Furthermore, in analogy with the  index
coding problem,  an interesting open question here is  whether linear codes
are always optimal, and whether there is any advantage to splitting
packets before linearly encoding them, i.e. whether vector linear
codes can lead to a lower number of transmissions than scalar linear
ones.

 Another important issue is to ensure ``fairness''  for all
clients. In this paper we were not concerned with the number of
transmissions each clients makes. In practice, clients have limited
energy resources and hence, it makes sense to find solutions where
the number of transmissions is as uniformly distributed among different clients as possible.